\documentclass{article}
\usepackage{romp,amssymb,amsfonts,amsmath,dsfont,amsthm}

\usepackage{hyperref}

\newenvironment{remark}[1][Remark]{\begin{trivlist}
\item[\hskip \labelsep {\bfseries #1}]}{\end{trivlist}}
\newenvironment{warning}[1][Warning]{\begin{trivlist}
\item[\hskip \labelsep {\bfseries #1}]}{\end{trivlist}}
\newenvironment{aside}[1][Aside:]{\begin{trivlist}
\item[\hskip \labelsep {\bfseries #1}]}{\end{trivlist}}
\newenvironment{statement}[1][Statement:]{\begin{trivlist}
\item[\hskip \labelsep {\bfseries #1}]}{\end{trivlist}}

\newcommand{\Dleft}{[\hspace{-1.5pt}[}
\newcommand{\Dright}{]\hspace{-1.5pt}]}
\newcommand{\SN}[1]{\Dleft #1 \Dright}

\DeclareMathOperator{\Vect}{Vect}

\DeclareMathOperator{\w}{w}

\title{ From $L_{\infty}$-algebroids to higher Schouten/Poisson structures }
\author{ Andrew James Bruce \\ (e-mail: andrewjames.bruce@physics.org) \\[2ex]
                      }

\begin{document}

\maketitle
\begin{abstract}
 We show that  $L_{\infty}$-algebroids, understood in terms of $Q$-manifolds can be described in terms of certain higher Schouten and Poisson structures on graded (super)manifolds. This generalises known constructions for Lie (super) algebras and Lie algebroids.
\end{abstract}


\section{Introduction and main results}\label{introduction}
\noindent  Recall that Lie algebroids \cite{Pradines1967} were originally  defined  as the triple $(E, [\bullet,\bullet], a )$, here $E$ is a vector bundle over the manifold $M$ equipped with a Lie bracket acting on the module of sections $\Gamma(E)$,  together with a vector bundle morphism called the anchor $a: E \rightarrow TM$.  The anchor and the Lie bracket satisfy the following

\begin{equation}\label{liealgebroidequations}
[u,fv] = a(u)f v \pm f[u,v],\hspace{15pt}
a([u,v]) = [a(u), a(v)],
\end{equation}

\noindent for all $u,v \in \Gamma(E)$ and $f \in C^{\infty}(M)$. To paraphrase this definition, a Lie algebroid is a vector bundle with the structure of a Lie algebra on the module of sections that can be represented by vector fields.\\

 \noindent Equivalently, a vector bundle  $E \rightarrow M$  is a Lie algebroid if there exists a weight one homological vector field on the total space of $\Pi E$, considered as a graded manifold \cite{Vaintrob:1997}. A supermanifold equipped with a homological vector field often denoted $Q$, that is an odd vector field that Lie supercommutes with itself, is known as a Q-manifold.  Note that  from the start we will consider all objects to be $\mathds{Z}_{2}$-graded, we will refer to this grading as (Grassmann) parity. Here $\Pi$ is the parity reversion functor, it shifts the parity of the fibre coordinates. The weight is provided by the assignment of weight zero to the base coordinates and   weight one to the fibre coordinates.  Generically the weight is completely  independent of the parity.   The homological condition on the vector field  encapsulates all the properties of Lie algebroids, namely equations (\ref{liealgebroidequations}). \\

\noindent What is slightly less well-know  is that the algebroid structure on $E \rightarrow M$ is also equivalent to

\begin{enumerate}
\item A weight minus one Schouten\footnote{Schouten  structures are also known as odd Poisson or Gerstenhaber structures. We will stick to the nomenclature Schouten following \cite{Voronov:1995}. They are the Grassmann odd analogue of Poisson structures.} structure on the total space of  $\Pi E^{*}$.
\item A weight minus one Poisson structure on the total space of $E^{*}$.
\end{enumerate}

\noindent It is important to note that the description of Lie algebroids  as certain Schouten and Poisson structures is in terms of functions on graded supermanifolds, as opposed to sections of vector bundles. Note that the linearity of these brackets in ``conventional language"  is replaced by a condition on the weight.  Moreover, the associated  Schouten and Poisson brackets satisfy a Leibnitz rule over the product of functions. For the case of a Lie (super) algebra the associated  brackets are known as the Lie--Schouten and  Lie--Poisson bracket \cite{Voronov:1995}. \\

\noindent We address the natural  question  ``\emph{is there is a similar construction for $L_{\infty}$-algebroids?}"\\

\noindent We understand an $L_{\infty}$-algebroid to be the Q-manifold $(\Pi E,Q)$, for a given  vector bundle $E \rightarrow M$. The homological vector field can be inhomogenous in weight.  A  notion of strictness, thought of as a compatibility condition between the Q-structure and the vector bundle structure can be employed. The two extreme examples of $L_{\infty}$-algebroids are Lie algebroids and $L_{\infty}$-algebras. \\

\noindent One can also describe $L_{\infty}$-algebroids in terms of an $L_{\infty}$-algebra on the module of sections $\Gamma(E)$ such that the ``higher anchors" arise in the Leibnitz rule.  In this work we will take the description of $L_{\infty}$-algebroids  in terms of Q-manifolds as the starting point. This is conceptually clear and fundamental in the constructions presented in this paper. Furthermore, it allows for a clear definition of morphisms between $L_{\infty}$-algebroids as morphisms of  Q-manifolds.   \\

 \noindent Thinking of Poisson and Schouten structures as functions on particular symplectic supermanifolds allows for very natural higher generalisations as outlined by Voronov \cite{voronov-2004,voronov-2005}. These higher structures are precisely what are required in passing from Lie algebroids to  $L_{\infty}$-algebroids.\\

 \noindent We state the main theorem (Theorem(\ref{theorem 1})) of this paper as the canonical construction of total weight one higher Schouten or higher Poisson structures on the total space of $\Pi E^{*}$ or $E^{*}$ respectively, given an $L_{\infty}$-algebroid $(\Pi E, Q)$. That is we associate with the homological field $Q \in \Vect(\Pi E)$ an odd  function $S \in C^{\infty}(T^{*}(\Pi E^{*}))$ such that $\{S,S \}_{T^{*}(\Pi E^{*})}=0$  and an even function $P \in C^{\infty}(\Pi T^{*}(E^{*}))$ such that $\SN{P,P}_{\Pi T^{*}(E^{*})}=0$. The brackets here are  canonical Poisson and Schouten--Nijenhuits brackets respectively.  By employing a bi-grading it is shown  that these structures can be assigned a total weight of one. \\

 \noindent The higher Schouten and higher Poisson structures are thought of as a higher order  generalistion of  the ``classical binary" structures.  For example, a higher Poisson structure on a supermanifold is the replacement of a Poisson bi-vector with an even parity, but otherwise inhomogenous multivector field. Associated with a higher Schouten/Poisson structure is a homotopy Schouten/Poisson algebra on the smooth functions over the supermanifold.  That is there is an \emph{$L_{\infty}$-algebra structure, suitably ``superised" such that the series of brackets satisfy a Leibnitz rule over the supercommutative product of functions}. See Voronov \cite{voronov-2004,voronov-2005}  and  Voronov \& Khudaverdian \cite{khudaverdian-2008} (also see de Azc$\acute{\textnormal{a}}$rraga et.al \cite{deAzcarraga:1996zk,deAzcarraga:1996jk}).\\

\noindent For the specific case of $L_{\infty}$-algebroids, the algebras of ``vector bundle multivectors" $C^{\infty}(\Pi E^{*})$ and ``vector bundle symmetric contravariant tensors" $C^{\infty}(E^{*})$ come equipped with homotopy Schouten and homotopy Poisson algebras respectively. Furthermore, as a direct corollary of  Theorem(\ref{theorem 1}) we see that given an arbitrary $L_{\infty}$-algebra one can associate directly a certain homotopy Schouten or equivalently a certain homotopy Poisson algebra. That is the constructions presented in this paper give  homotopy versions of the Lie--Schouten and Lie--Poisson brackets.  \\

\noindent  $L_{\infty}$-algebroids appear quite directly when considering higher Poisson and higher Schouten structures over supermanifolds \cite{khudaverdian-2008}, as well as on Lie algebroids \cite{Bruce-2010}. In particular they arise when generalising  the Lie algebroid structure on  $T^{*}M$ for a given Poisson manifold to the higher/homotopy versions. In part the work presented here goes towards completing  the general framework found in the study of higher Poisson and Schouten structures.  \\

\noindent However, it must be noted that the notion of an $L_{\infty}$-algebroid employed here is not the most general  one could consider. More general graded manifolds and homological vector fields on them, that is  ``differential graded manifolds"  would represent a wider definition of an $L_{\infty}$-algebroid than employed here.  All the graded structures encountered here will have their origin in vector bundle and double vector bundle structures.  Examples of differential graded manifolds, can for example be found lying behind the BV-antifield formalism \cite{Batalin:1981jr,Batalin:1984jr} and the BFV formalism \cite{Fradkin1975,Batalin1977}. Indeed quantum field theory provides a rich source of graded structures and will continue to provide much inspiration to mathematicians.  \\

 \noindent Graded geometry has provided a powerful setting to discus various geometric constructions. Works along these lines  include \cite{Alexandrov:1995kv,Bruce-2010,roytenberg-1999,Roytenberg:2001,voronov-2004,Voronov:2001qf,Severa-2001} as well as many others. \\

 \noindent This section continues with  a brief outline of $L_{\infty}$-algebras and higher derived brackets as needed later.  Here we will fix some nomenclature, notation and conventions. In Section(\ref{Linftyalgebroids}) we recall some basic facts about graded manifolds and define $L_{\infty}$-algebroids.  In Section(\ref{main theorem}) we state and prove the main theorem of this paper, Theorem(\ref{theorem 1}). We also include a few explicit and simple examples to illustrate the theorem in Section(\ref{simpleexamples}).  In Section(\ref{concluding remarks}) we end with  few concluding remarks. A short appendix presenting some  lemmas on canonical double vector bundle morphisms is included. \\

\noindent \textbf{Preliminaries} \\
 \noindent All vector spaces and algebras will be $\mathds{Z}_{2}$-graded. The reason for this lies in physics, where it is necessary to employ such a grading when wanting to describe fermions and/or ghosts.  Generally  in accordance with  ``supermathematics"  we will omit the prefix \emph{super}. By \emph{manifold} we will mean a \emph{smooth supermanifold}. We denote the Grassmann parity of an object by \emph{tilde}: $\widetilde{A} \in \mathds{Z}_{2}$. By \emph{even} or \emph{odd} we will be referring explicitly to the Grassmann parity and not to any extra weight(s). \\

\noindent A \emph{Poisson} $(\varepsilon = 0)$  or \emph{Schouten} $(\varepsilon = 1)$ \emph{algebra} is understood as a vector space $A$ with a bilinear associative multiplication and a bilinear operation $\{ ,\}: A \otimes A \rightarrow A$ such that:
\begin{list}{}
\item \textbf{Grading} $\widetilde{\{a,b \}_{\varepsilon}} = \widetilde{a} + \widetilde{b} + \varepsilon$
\item \textbf{Skewsymmetry} $\{a,b\}_{\varepsilon} = -(-1)^{(\tilde{a}+ \varepsilon)(\tilde{b}+ \varepsilon)} \{b,a \}_{\varepsilon}$
\item \textbf{Leibnitz Rule} $\{a,bc \}_{\varepsilon} = \{a,b \}_{\varepsilon}c + (-1)^{(\tilde{a} + \varepsilon)\tilde{b}} b \{a,c \}_{\varepsilon}$
\item \textbf{Jacobi Identity} $\sum_{cyclic\: a,b,c} (-1)^{(\tilde{a}+ \varepsilon)(\tilde{c}+ \varepsilon)}\{a,\{b,c\}_{\varepsilon}  \}_{\varepsilon}= 0$
\end{list} \vspace{10pt}
\noindent for all homogenous elements $a,b,c \in A$.\\

 \noindent A manifold $M$ such that $C^{\infty}(M)$ is a Poisson/Schouten algebra is known as a \emph{Poisson/Schouten manifold}. As the Poisson/Schouten brackets are biderivations over the functions they are  specified by contravariant tensor fields of rank two. A \emph{Poisson structure} on a manifold $M$ is understood as a bi-vector field $P \in C^{\infty}(\Pi T^{*}M)$ (quadratic in fibre coordinates), such that $\SN{P,P} = 0$. Here the brackets are the canonical Schouten brackets on $\Pi T^{*}M$ also known a the Schouten--Nijenhuis bracket.  The associated Poisson bracket is given by $\{f,g \}_{P} = (-1)^{\widetilde{f}+1}\SN{\SN{P,f},g}$,  with $f,g \in C^{\infty}(M)$. Similarly, a \emph{Schouten structure} on a manifold $M$ is an odd symmetric tensor field $S \in C^{\infty}(T^{*}M)$ quadratic in the fibre coordinates such that $\{S,S\}=0$. The associated Schouten bracket us given by $\SN{f,g}_{S}=(-1)^{\widetilde{f}+1} \{ \{ S,f  \},g  \}$, with  $f,g \in C^{\infty}(M)$. Note that non-trivial Schouten structures cannot exist on pure even manifolds.  The Jacobi identities on the brackets are equivalent to the self-commutating conditions of the structures.  \\

\noindent  We closely follow Voronov \cite{voronov-2004} in conventions concerning $L_{\infty}$-algebras. A  vector space $V = V_{0}\oplus V_{1}$ endowed with a sequence of odd n-linear operators of $n \geq 0$ (which we denote as $(\bullet,\cdots,\bullet) $) is said to be an $L_{\infty}$-algebra (c.f. \cite{Lada:1994mn,Lada:1992wc}) if
\begin{enumerate}
\item The operators are symmetric
\begin{equation}
(a_{1}, a_{2}, \cdots, a_{i},a_{j}, \cdots , a_{n}) = (-1)^{\widetilde{a}_{i}\widetilde{a}_{j}}(a_{1}, a_{2}, \cdots, a_{j},a_{i}, \cdots , a_{n}).
\end{equation}
\item The generalised Jacobi identities
\begin{equation}
\sum_{k+l=n-1} \sum_{(k,l)-\textnormal{unshuffels}}(-1)^{\epsilon}\left( (a_{\sigma(1)}, \cdots , a_{\sigma(k)}), a_{\sigma(k+1)}, \cdots, a_{\sigma(k+l)} \right)=0
\end{equation}
hold for all $n \geq 1$. Here $(-1)^{\epsilon}$ is a sign that arises due to the exchange of homogenous elements $a_{i} \in V$. Recall that a $(k,l)$-unshuffle is a permutation of the indices $1, 2, \cdots k+l$ such that $\sigma(1) < \cdots < \sigma(k)$ and $\sigma(k+1) < \cdots < \sigma(k+l)$. The LHS of the above are referred to as Jacobiators.
\end{enumerate}

\noindent It must be noted that the above definitions are shifted as compared to the original definitions of Lada \& Stasheff.  Specifically, if $V = \Pi U$ is an $L_{\infty}$-algebra (as above) then we have a series of brackets on $U$ that are skew-symmetric and  even/odd for an even/odd number of arguments. Let  $x_{i} \in U$ and  we define the brackets on $U$ viz
\begin{equation}
\Pi \{x_{1}, \cdots , x_{n} \} = (-1)^{(\widetilde{x}_{1}(n-1) + \widetilde{x}_{2}(n-2)+ \cdots + \widetilde{x}_{n-1} +1)}(\Pi x_{1}, \cdots , \Pi x_{n}).
\end{equation}
One may call $V = \Pi U$ an $L_{\infty}$-antialgebra.  However, we will refer to the bracket structures on $V$ and $U$ as  $L_{\infty}$-algebras keeping in mind the above identification.\\

\begin{warning}
There is plenty of room here over the assignments of gradings and symmetries. We prefer to work in the ``super-setting". It must also be remarked that in most applications the zero bracket vanishes identically. In such cases we say that the $L_{\infty}$-algebra is \emph{strict}. In the literature $L_{\infty}$-algebras with a non-vanishing zero bracket are called ``weak", ``with background" or ``curved". By default, we will include a non-vanishing zero bracket unless otherwise stated.
\end{warning}

\begin{definition}{Definition}\label{homotopy Schouten}
A homotopy Schouten algebra is a commutative, associative, unital algebra $\mathcal{A}$ equipped with an $L_{\infty}$-algebra structure such that the odd $n$-multilinear operations known as higher Schouten brackets, are multiderivations over the product:
\begin{eqnarray}
(a_{1}, a_{2}, \cdots a_{r-1}, a_{r}a_{r+1}) &=& (a_{1}, a_{2}, \cdots a_{r-1}, a_{r}) a_{r+1}\\
 \nonumber &+& (-1)^{\widetilde{a_{r}}(\widetilde{a_{1}} +\widetilde{a_{2}} + \cdots + \widetilde{a_{r-1}} +1)}a_{r}(a_{1}, a_{2}, \cdots a_{r-1},a_{r+1}),
\end{eqnarray}
with $a_{I} \in \mathcal{A}$.
\end{definition}

\noindent In order to define a homotopy Poisson algebra one needs to consider a shift in parity to keep inline with our conventions. Up to this shift, the definition carries over directly.

\begin{definition}{Definition}\label{homotopy Poisson}
A homotopy Poisson algebra is a commutative, associative, unital algebra $\mathcal{A}$ equipped with an $L_{\infty}$-algebra structure such that the  $n$-multilinear operations known as higher Poisson brackets (even/odd for even/odd number of arguments), are multiderivations over the product:
\begin{eqnarray}
\{a_{1}, a_{2}, \cdots a_{r-1}, a_{r}a_{r+1}\} &=& \{a_{1}, a_{2}, \cdots a_{r-1}, a_{r}\} a_{r+1}\\
 \nonumber &+& (-1)^{\widetilde{a_{r}}(\widetilde{a_{1}} +\widetilde{a_{2}} + \cdots + \widetilde{a_{r-1}} +r)}a_{r}\{a_{1}, a_{2}, \cdots a_{r-1},a_{r+1}\},
\end{eqnarray}
with $a_{I} \in \mathcal{A}$.
\end{definition}

\noindent Following Voronov \cite{voronov-2004} it is known how to construct a series of brackets from the ``initial data"-- $\left(\mathcal{L},\pi, \Delta \right)$. Here $\mathcal{L}$ is a Lie (super)algebra equipped with a projector ($\pi^{2} = \pi$) onto an abelian subalgebra satisfying the distributivity rule $\pi[a,b] = \pi[\pi a,b] + \pi[a, \pi b]$ for all $a,b \in \mathcal{L}$.  Given an element $\Delta \in \mathcal{L}$ a series of brackets on the abelian subalgebra, $V \subset \mathcal{L}$ is defined as

\begin{equation}
(a_{1},a_{2}, \cdots,a_{n}) = \pi[\cdots[[[\Delta, a_{1} ],a_{2}],\cdots a_{n}],
\end{equation}

\noindent with $a_{i}$ in $V$. The zero bracket is defined as
\begin{equation}
(\emptyset) = \pi \Delta.
\end{equation}
\noindent Such brackets have the same parity as $\Delta$ and are symmetric.  The series of brackets is referred to as higher derived brackets generated by $\Delta$. A theorem due to Voronov states that  for an odd generator $\Delta \in \mathcal{L}$ the n-th Jacobiator is given by the n-th higher derived bracket generated by $\Delta^{2}$.

\begin{equation}
J^{n}(a_{1},a_{2}, \cdots,a_{n}) = \pi [\cdots[[[\Delta^{2}, a_{1} ],a_{2}],\cdots a_{n}].
\end{equation}

\noindent In particular we have that if $\Delta^{2} =0$ then the series of higher derived brackets is an $L_{\infty}$-algebra. Note that if $\pi \Delta =0$ then the $L_{\infty}$-algebra is strict. \\

\begin{definition}{Definition}\label{higher Poisson structure}
Let $M$ be a manifold. An even multivector field $P \in C^{\infty}(\Pi T^{*}M)$ is said to be a higher Poisson structure if and only if $\SN{P,P}=0$, where the bracket is the canonical Schouten--Nijenhuist bracket on $\Pi T^{*}M$.
\end{definition}

\noindent Via Voronov's higher derived bracket formalism one obtains a homotopy Poisson algebra on $C^{\infty}(M)$ when $M$ is equipped with a higher Poisson structure. The brackets being given by

\begin{equation}
\{f_{1}, f_{2}, \cdots, f_{r}   \}_{P} =  (-1)^{\widetilde{f}_{1} (r-1) + \widetilde{f}_{2}(r-2) + \cdots + \widetilde{f}_{r-1} +r    } \left .\SN{ \cdots  \SN{ \SN{P, f_{1}}, f_{2}  }, \cdots , f_{r}   }\right|_{M},
\end{equation}

\noindent where $f_{I} \in C^{\infty}(M)$. Note the above sign factor ensures that the higher Poisson brackets are skewsymmetric. It is possible to ignore this sign factor and work with antisymmetric brackets.   \\

\begin{definition}{Definition}\label{higher Schouten structure}
Let $M$ be a manifold. An odd function $S \in C^{\infty}(T^{*}M)$ is said to be a higher Schouten structure if and only if $\{S,S\}=0 $, where the bracket is the canonical Poisson bracket on $T^{*}M$.
\end{definition}

\noindent One obtains a homotopy Schouten algebra on $C^{\infty}(M)$ when $M$ is equipped with a higher Schouten structure. The brackets being given by
\begin{equation}
(f_{1}, f_{2}, \cdots , f_{r})_{S} = \left.\{  \cdots \{\{S,f_{1}\}, f_{2}  \}, \cdots , f_{r}\}    \right|_{M},
\end{equation}
\noindent where $f_{I} \in C^{\infty}(M)$. \\

\begin{definition}{Definition}\label{homological vector}
An odd vector field $Q \in \Vect(M)$ that ``squares to zero", that is   $[Q,Q]= 2 Q^{2}=0$ shall be known as a homological vector field.
\end{definition}

\begin{definition}{Definition}\label{Qmanifold}
A manifold equipped with a homological vector field shall be  known as a Q-manifold.
\end{definition}

\noindent In fact, all $L_{\infty}$-algebras can be understood in terms of formal Q-manifolds. If set $V = \Pi U$ (as vector spaces), then we can consider elements of $\Pi U$ as being ``constant valued" vector fields

\begin{eqnarray}
i: \Pi U & \hookrightarrow & \Vect(\Pi U)\\
\nonumber a = a^{\alpha}s_{\alpha} &\rightarrow& a^{\alpha}\frac{\partial}{\partial \xi^{\alpha}},
\end{eqnarray}

\noindent where we have picked an ``odd basis", $\widetilde{s_{\alpha}} = (\widetilde{\alpha}+1)$ and local coordinates $\{\xi^{\alpha}\}$ on $\Pi U$ considered as a formal supermanifold. An $L_{\infty}$-algebra  is encoded in a homological vector field of arbitrary weight (assign weight one to the linear coordinates)

\begin{equation}
Q = \left( Q^{\delta}_{0} + \xi^{\alpha}Q_{\alpha}^{\delta} + \frac{1}{2!}\xi^{\alpha}\xi^{\beta} Q_{\beta \alpha}^{\delta} + \frac{1}{3!} \xi^{\alpha}\xi^{\beta} \xi^{\gamma} Q_{\gamma \beta \alpha}^{\delta} + \cdots  \right)\frac{\partial}{\partial \xi^{\delta}}.
\end{equation}

\noindent The series of brackets are then given by

\begin{equation}
(a_{1}, a_{2}, \cdots , a_{r}) = \pi_{0} \left( \left[  \cdots [\cdots [Q, a_{1}], a_{2}], \cdots , a_{r} \right]\right),
\end{equation}

\noindent where the projector $\pi_{0}$ here is the evaluation at the origin. In terms of a local basis the brackets can be expressed as

\begin{equation}
(s_{\alpha_{1}}, s_{\alpha_{2}}, \cdots , s_{\alpha_{r}} ) = (-1)^{(\sum_{i=1}^{r}\widetilde{\alpha}_{i}) } Q^{\beta}_{\alpha_{1} \alpha_{2} \cdots \alpha_{r}} s_{\beta}.
\end{equation}

\noindent A little more explicitly the first few brackets are given by

\begin{eqnarray}
\nonumber (\emptyset) = Q_{0}^{\delta}s_{\delta}, && (s_{\alpha}) = (-1)^{\widetilde{\alpha}} Q_{\alpha}^{\delta}s_{\delta},\\
\nonumber (s_{\alpha}, s_{\beta}) = (-1)^{\widetilde{\alpha} + \widetilde{\beta}}Q_{\alpha \beta}^{\delta}s_{\delta}, && (s_{\alpha}, s_{\beta}, s_{\gamma}) = (-1)^{\widetilde{\alpha} + \widetilde{\beta}+  \widetilde{\gamma} }Q_{\alpha \beta\gamma}^{\delta}s_{\delta}.
\end{eqnarray}

\begin{remark}
It is also true   that $A_{\infty}$-algebras and $C_{\infty}$-algebras can be understood in terms of formal Q-manifolds. As we will have no use for them here we will not elaborate further.
\end{remark}

\begin{warning} The notion of homotopy Schouten and homotopy Poisson algebra used in this work is far more restrictive than found elsewhere in the literature. We will make no use of the theory of (pr)operads in our constructions, \cite{galvezcarrillo-2009,Ginzburg-1994}.  Specifically, the homotopy Poisson and Schouten algebras defined here  are not the cofibrant resolution of the appropriate operads. Only the Jacobi identity has been weakened \emph{up to homotopy}. Furthermore, these notions can be formulated in the $\mathds{Z}$-graded setting. For example see Tamarkin \& Tsygan \cite{Tamarkin2000}, Cattaneo \& Felder \cite{Cattaneo2007} and Mehta \cite{Mehta-2010}.   However, the  $\mathds{Z}_{2}$-graded notions used in this paper seem very natural for supergeometry and suit the purposes of this work very well.
\end{warning}

\section{Graded manifolds and $L_{\infty}$-algebroids}\label{Linftyalgebroids}

\noindent  Recall the definition of a (multi)graded manifold as a manifold $\mathcal{M}$, equipped with a privileged class of atlases where the coordinates are assigned weights taking values in $\mathds{Z}^{n}$  ($n \in \mathds{N}$) and the coordinate transformations are polynomial in coordinates with nonzero weights respecting the  weights, see for example \cite{Grabowski2009,Roytenberg:2001,Voronov:2001qf,Severa-2001}. Generally the weight  will be independent of the Grassmann parity. Moreover, any sign factors that arise will be due to the Grassmann parity and we do not include any possible extra signs due to the weight(s). In simpler terms, we have a manifold equipped with a  distinguished class of charts and diffeomorphisms between them respecting the $\mathds{Z}_{2}$-grading as well as the additional $\mathds{Z}^{n}$-grading. These gradings then pass over to geometric objects on graded manifolds.\\

 \noindent Let us employ local coordinates $\{ x^{A}\}$ on an arbitrary graded manifold $\mathcal{M}$. We  will use the notation  $\w(x^{A}) = (\w_{1}(x^{A}), \w_{2}(x^{A}) \, \cdots , \w_{n}(x^{A})) \in \mathds{Z}^{n}$  for  the weight. One can then pass to a total weight $\#(x^{A}) = \sum_{i=1}^{n}w_{i}(x^{A})$. In this work we will only require up to a bi-weight. That is at most the weights will take their values in $\mathds{Z}^2$.\\

\noindent A vector bundle structure $E \rightarrow M$ is  equivalent to  the total space of the anti-vector bundle $\Pi E$ having a certain graded structure, under the assumption of no external weighted parameters being employed. To be more specific, let us employ natural coordinates $\{x^{A}, \xi^{\alpha} \}$ on $\Pi E$.  We assume $M$ is just a manifold as opposed to a graded manifold. The  parities being given by $\widetilde{x}^{A} = \widetilde{A}$ and $\widetilde{\xi}^{\alpha}= \widetilde{\alpha}+1$. Furthermore, let us assign the weights $\w(x^{A}) = 0 $ and $\w(\xi^{\alpha}) = 1$. Then the admissible changes of coordinates are necessarily of the form $\overline{x}^{A} = \overline{x}^{A}(x)$ and $\overline{\xi}^{\alpha} = \xi^{\beta}T_{\beta}^{\:\: \alpha}(x) $. Thus, we demonstrated this assertion. The zero section of $E \rightarrow M$ is identified with the zero weight part of $\Pi E$. More correctly, $C^{\infty}(M)\subset C^{\infty}(\Pi E)$ as the zero weight subalgebra.   Note, that other choices in weight are also perfectly valid.\\

\begin{definition}{Definition}
A vector bundle $E \rightarrow M$ is said to have an $L_{\infty}$-algebroid structure if there exists  a homological vector field $Q \in \Vect(\Pi E)$. That is, the total space of the anti-vector bundle $\Pi E$ is a Q-manifold. The pair $(\Pi E, Q )$ will be known as an $L_{\infty}$-algebroid.
\end{definition}

\noindent Note that there is no condition on the weight of the homological vector field in this definition. Recall that for a Lie algebroid the weight of the homological vector field is one.\\

\noindent  Throughout this work  the Q-manifold $(\Pi E, Q)$ is considered as the \emph{primary object}. Morphisms of $L_{\infty}$-algebroids are understood as  morphisms in the category of (graded) Q-manifolds.\\

\noindent If we employ natural local coordinates  $\{x^{A}, \xi^{\alpha}  \}$ the homological vector field is of the form:

\begin{eqnarray}\label{homologicalfieldLieinfity}
Q &=& Q^{A}(x, \xi) \frac{\partial}{\partial x^{A}} + Q^{\alpha}(x, \xi)\frac{\partial }{\partial \xi^{\alpha}}\\
\nonumber &=& \left(Q^{A}(x) + \xi^{\alpha}Q_{\alpha}^{A}(x) + \frac{1}{2!} \xi^{\alpha}\xi^{\beta} Q_{\beta \alpha}^{A}(x) +\cdots    \right)\frac{\partial}{\partial x^{A}}\\
\nonumber &+& \left( Q^{\alpha}(x) + \xi^{\beta}Q_{\beta}^{\alpha}(x)  + \frac{1}{2!} \xi^{\beta}\xi^{\gamma}Q_{\gamma \beta}^{\alpha}(x) +  \cdots \right)\frac{\partial }{\partial \xi^{\alpha}}.
\end{eqnarray}

\noindent Recall that the algebra of smooth functions on a graded manifold is understood as the formal completion of  the polynomial algebra in  weighted coordinates. Thus, the components of the homological vector field may be  understood very formally. Alternatively, more  concretely one could  consider only finite order polynomials. This leads to the notion a Lie $n$-algebroid as an $L_{\infty}$-algebroid whose homological vector field concentrated in weight  up to $n-1$. We will not dwell on this. \\

\begin{definition}{Definition} An $L_{\infty}$-algebroid  $(\Pi E, Q)$ is said to be a strict $L_{\infty}$-algebroid if and only of the homological vector field along the ``zero section" $M \subset \Pi E$ is a homological vector field on $M$.
\end{definition}

\noindent   In local coordinates this is the statement that  $Q^{\alpha}(x)=0$.  In a more invariant language, an $L_{\infty}$-algebroid is strict if and only if the homological vector field $Q \in \Vect(\Pi E)$ is the formal sum  of strictly non-negative weight vector fields: $Q  = \sum_{i=0}^{\infty}Q_{i}$. Such a condition automatically holds for Lie algebroids and reproduces the notion of a strict $L_{\infty}$-algebra thought of as an $L_{\infty}$-algebroid over a ``point". \\

\noindent Throughout this work we will not insist upon strictness \emph{a priori}, though it will feature later when discussing higher Schouten and higher Poisson structures associated with $L_{\infty}$-algebroids.\\

\begin{small}
\begin{aside}
An $L_{\infty}$-algebroid can also be understood as an $L_{\infty}$-algebra on the module of sections $\Gamma(E)$ such that the higher anchors arise in terms of the Leibnitz rule. A little more specifically (being  quite lax about signs) one has
\begin{equation}
[u_{1}, \cdots u_{r}, f \: u_{r+1}] = a(u_{1}, \cdots, u_{r})[f] u_{r+1} \pm f\: [u_{1}, \cdots u_{r}, u_{r+1}],
\end{equation}
\noindent with $u_{I} \in \Gamma(E)$ and $f \in C^{\infty}(M)$. In terms of a basis $s_{\alpha}$ ($\widetilde{s}_{\alpha}= \widetilde{\alpha}$) the anchors and brackets are given by:
\begin{subequations}
\begin{eqnarray}
a(s_{\alpha_{1}}, \cdots , s_{\alpha_{r}}) &=& \pm Q^{A}_{\alpha_{1} \cdots \alpha_{r}}\frac{\partial}{\partial x^{A}},\\
\left[s_{\alpha_{1}}, \cdots, s_{\alpha_{r}}\right] &=& \pm Q^{\beta}_{\alpha_{1}  \cdots \alpha_{r}}s_{\beta}.
\end{eqnarray}
\end{subequations}
\noindent The condition of strictness on the homological vector field $Q\in \Vect(\Pi E)$ is identical to the $L_{\infty}$-algebra on the module of sections being strict. That is there is no zero-bracket. However, there is still (potentially) a zero-anchor. We believe that the formulation in terms of Q-manifolds is clearer and more powerful than considering the module of sections.
\end{aside}
\end{small}

\section{The higher Schouten and Poisson structures associated with an $L_{\infty}$-algebroid}\label{main theorem}

\noindent We are now in a position to state and prove the main theorem of this paper.

\begin{theorem}{Theorem}\label{theorem 1}
An $L_{\infty}$-algebroid  $(\Pi E , Q)$ is equivalent to:
\begin{enumerate}
\item A  higher Schouten structure $S\in C^{\infty}(T^{*}(\Pi E^{*}))$ of total weight one.
\item A  higher Poisson structure  $P \in C^{\infty}(\Pi T^{*}(E^{*}))$ of total weight one.
\end{enumerate}
\end{theorem}

\begin{proof}
\noindent Let us employ natural coordinates $\{x^{A}, \xi^{\alpha} \}$ on $\Pi E$. Let the homological vector field defining the $L_{\infty}$-algebroid be  given by $Q = Q^{A}(x, \xi) \frac{\partial}{\partial x^{A}} + Q^{\alpha}(x, \xi)\frac{\partial }{\partial \xi^{\alpha}} \in \Vect(\Pi E)$ .\\
\begin{enumerate}
\item Let us employ natural local coordinates $\{x^{A}, \eta_{\alpha}, p_{A} , \pi^{\alpha}\}$ and $\{x^{A}, \xi^{\alpha}, p_{A}, \pi_{\alpha} \}$ on $T^{*}(\Pi E^{*})$ and $T^{*}(\Pi E)$ respectively, see Appendix(\ref{A1}). The bi-weights are assigned as $\w(x^{A}) =(0,0)$, $\w(\eta_{\alpha}) =(1,0)$, $\w(p_{A}) =(0,1)$, $\w(\pi^{\alpha}) =(-1,1)$ , $\w(\xi^{\alpha}) =(-1,1)$, $\w(\pi_{\alpha}) =(1,0)$. Note, these weights are compatible with the double vector bundle structures. Then taking the even principle symbol\footnote{see for example H$\ddot{\textnormal{o}}$rmander \cite{Hormander:1985III}.}  $\frac{\partial }{\partial x^{A}} \rightarrow p_{A}$, $\frac{\partial}{\partial \xi^{\alpha}} \rightarrow \pi_{\alpha}$  of the homological vector field gives:
    \begin{equation}
    \sigma Q = Q^{A}(x, \xi) p_{A}  + Q^{\alpha}(x, \xi) \pi_{\alpha} \in C^{\infty}(T^{*}(\Pi E)).
    \end{equation}
    \noindent It is well-know that the even  principle symbol maps commutators of vector fields to canonical Poisson brackets. This can very easily be directly  verified and directly follows from the definition of the principle symbol. Thus,
    \begin{equation}
    \sigma[Q,Q] =  \{\sigma Q, \sigma Q  \}_{T^{*}(\Pi E)} = 0.
    \end{equation}
    \noindent Then use the canonical double vector bundle  morphism (see Appendix(\ref{A1}) and/or \cite{Bruce-2010,mackenzie-2002,Voronov:2001qf}) \newline  $R : T^{*}(\Pi E^{*}) \rightarrow T^{*}(\Pi E)$ given by $R^{*}(\pi_{\alpha}) = \eta_{\alpha}$ and $R^{*}(\xi^{\alpha}) =  (-1)^{\widetilde{\alpha}}\pi^{\alpha}$ to define
    \begin{equation}
    S = (R^{-1})^{*}(\sigma Q) = Q^{A}(x, \pi)p_{A} + Q^{\alpha}(x, \pi) \eta_{\alpha}\in C^{\infty}(T^{*}(\Pi E^{*})),
    \end{equation}
    \noindent where we have used the shorthand notiation $Q^{A}(x, \pi) = (R^{-1})^{*}Q^{A}(x, \pi)$ and  \newline $Q^{\alpha}(x, \pi) = (R^{-1})^{*}Q^{\alpha}(x, \pi)$. In essence this is just the change of variables \newline $\pi_{\alpha} \rightarrow \eta_{\alpha}$ and $\xi^{\alpha} \rightarrow (-1)^{\widetilde{\alpha}} \pi^{\alpha}$ in the algebra of weighted polynomials. The condition $\{S,S \}_{T^{*}(\Pi E^{*})}=0$ follows from the fact that the canonical double vector bundle morphism is a symplectomorphism, see Lemma(\ref{A1}). Thus, $S$ is a higher Schouten structure on the total space of $\Pi E^{*}$ see Def.(\ref{higher Schouten structure}). Furthermore, it is clear that $\#(S) = 1$ by inspection.
\item Let is employ natural local coordinates $\{x^{A}, e_{\alpha}, x^{*}_{A}, e_{*}^{\alpha}  \}$ and $\{x^{A}, \xi^{\alpha}, x^{*}_{A}, \xi^{*}_{\alpha} \}$ on $\Pi T^{*}(E^{*})$ and $\Pi T^{*}(\Pi E)$ respectively, see Appendix(\ref{A2}). The bi-weights are assigned as $\w(x^{A}) =(0,0)$, $\w(e_{\alpha}) =(1,0)$, $\w(x^{*}_{A}) =(0,1)$, $\w(e_{*}^{\alpha}) =(-1,1)$, $\w(\xi^{\alpha}) =(-1,1)$, $\w(\xi_{\alpha}^{*}) =(1,0)$. Note, these weights are compatible with the double vector bundle structures. Then taking the odd principle symbol (a.k.a. odd isomorphism \cite{Voronov:1992}) $\frac{\partial}{\partial x^{A}} \rightarrow x^{*}_{A}$, $\frac{\partial}{\partial \xi^{\alpha}} \rightarrow \xi^{*}_{\alpha}$ of $Q$  gives:
    \begin{equation}
    \varsigma Q = Q^{A}(x, \xi) x^{*}_{A}  + Q^{\alpha}(x, \xi) \xi^{*}_{\alpha} \in C^{\infty}(\Pi T^{*}(\Pi E)).
    \end{equation}
    \noindent The odd principle symbol maps commutators of vector fields to canonical Schouten(--Nijenhuist) brackets. This can be easily and directly varified. Thus,
    \begin{equation}
    \varsigma[Q,Q] = \SN{\varsigma Q, \varsigma Q}_{\Pi T^{*}(\Pi E)} = 0.
    \end{equation}
    \noindent Then use the canonical double vector bundle morphism (see Appendix(\ref{A2}) and/or \cite{Bruce-2010})\newline $R : \Pi T^{*}(E^{*}) \rightarrow \Pi T^{*}(\Pi E)$ given by $R^{*}(\xi^{*}_{\alpha}) = - e_{\alpha}$ and $R^{*}(\xi^{\alpha}) = e_{*}^{\alpha}$ to define
    \begin{equation}
    P = (R^{-1})^{*} (\varsigma Q) = Q^{A}(x, e_{*})x_{A}^{*} - Q^{\alpha}(x, e_{*})e_{\alpha} \in C^{\infty}(\Pi T^{*}(E^{*})),
   \end{equation}
    \noindent we have used the shorthand  notiation $Q^{A}(x, e_{*}) = (R^{-1})^{*}Q^{A}(x, e_{*})$ and  \newline $Q^{\alpha}(x, e_{*}) = (R^{-1})^{*}Q^{\alpha}(x, e_{*})$. In essence this is just the change of variables $\xi_{\alpha}^{*} \rightarrow -e_{\alpha}$ and $\xi^{\alpha} \rightarrow e_{*}^{\alpha}$ in the algebra of weighted polynomials.  The condition $\SN{P,P }_{\Pi T^{*}(E^{*})} =0$ follows from the fact that the canonical double vector bundle morphism is a symplectomorphism, see Lemma(\ref{A2}). Thus $P$ is a higher Poisson structure on $E^{*}$, see Def.(\ref{higher Poisson structure}). Furthermore, it is clear that $\#(P) = 1$ by inspection.
\end{enumerate}
\end{proof}

\begin{remark}
Generally a Q-manifold $(\mathcal{M}, Q_{\mathcal{M}})$ (possibly in the category of graded manifolds) can  be considered as a higher Schouten or higher Poisson manifold of ``order one". That is the associated $L_{\infty}$-algebras on $C^{\infty}(\mathcal{M})$ consist of a single one-bracket. This is implemented  via $Q_{\mathcal{M}} \rightsquigarrow S_{\mathcal{M}} =\sigma Q_{\mathcal{M}}$ or $Q_{\mathcal{M}} \rightsquigarrow P_{\mathcal{M}} =\varsigma Q_{\mathcal{M}}$. Note that in fact any ``order one" Schouten or Poisson structure is equivalent to a homological vector field.  For the case at hand the ``higher order structure"  is in some sense  \emph{induced} by the canonical double vector bundle morphisms.
\end{remark}

\noindent Let us examine the local expressions in a little more detail.  Explicitly, if the homological vector field is formally given by:

\begin{eqnarray}
Q &=& \sum_{r=0}^{\infty} \left(\frac{1}{r!} \xi^{\alpha_{1}} \xi^{\alpha_{2}} \cdots \xi^{\alpha_{r}} Q_{\alpha_{r} \cdots \alpha_{2} \alpha_{1}}^{A}(x)\right)\frac{\partial}{\partial x^{A}}  \\
 \nonumber &+&  \sum_{r=0}^{\infty}\left(\frac{1}{r!} \xi^{\alpha_{1}} \xi^{\alpha_{2}} \cdots \xi^{\alpha_{r}}  Q_{\alpha_{r} \cdots \alpha_{2} \alpha_{1}}^{\beta}(x)\right)\frac{\partial}{\partial \xi^{\beta}},
\end{eqnarray}

\noindent then we have:
\begin{subequations}
\begin{eqnarray}
S &=& \sum_{r=0}^{\infty}\left((-1)^{\widetilde{\alpha}_{1} + \cdots + \widetilde{\alpha}_{r}}\frac{1}{r!} \pi^{\alpha_{1}} \pi^{\alpha_{2}} \cdots \pi^{\alpha_{r}} Q_{\alpha_{r} \cdots \alpha_{2} \alpha_{1}}^{A}(x)\right) p_{A} \\
 \nonumber &+& \sum_{r=0}^{\infty}\left( (-1)^{\widetilde{\alpha}_{1} + \cdots + \widetilde{\alpha}_{r} }\frac{1}{r!} \pi^{\alpha_{1}} \pi^{\alpha_{2}} \cdots \pi^{\alpha_{r}}  Q_{\alpha_{r} \cdots \alpha_{2} \alpha_{1}}^{\beta}(x)\right)\eta_{\beta},\\
P &=& \sum_{r=0}^{\infty}\left(\frac{1}{r!} e_{*}^{\alpha_{1}} e_{*}^{\alpha_{2}} \cdots e_{*}^{\alpha_{r}} Q_{\alpha_{r} \cdots \alpha_{2} \alpha_{1}}^{A}(x)\right)x^{*}_{A}\\
\nonumber  &-& \sum_{r=0}^{\infty}\left(\frac{1}{r!} e_{*}^{\alpha_{1}} e_{*}^{\alpha_{2}} \cdots e_{*}^{\alpha_{r}}
 Q_{\alpha_{r} \cdots \alpha_{2} \alpha_{1}}^{\beta}(x)\right)e_{\beta}.
\end{eqnarray}
\end{subequations}

\begin{remark}
The higher Schouten and higher Poisson structures associated with an $L_{\infty}$-algebroid are far from being the most general structures that could be studied. The total weight one ensures the higher structures have the correct ``linearity". This opens up another possible generalistion of Lie algebroids as objects dual to more  general higher Schouten  and higher Poisson structures on the total spaces of $\Pi E^{*}$ and $E^{*}$ respectively.
\end{remark}

\noindent These structures provide the  algebras $C^{\infty}(\Pi E^{*})$ and $C^{\infty}(E^{*})$  with a series of brackets that form homotopy Schouten and homotopy Poisson algebras  respectively. That is $L_{\infty}$-algebras in the sense of Lada \& Stasheff \cite{Lada:1992wc} suitably ``superised" such that the brackets are multiderivations over the commutative product of functions, see Def.(\ref{homotopy Poisson}) and Def.(\ref{homotopy Schouten}).  \\

\noindent  Note that the higher structures can be presented as  the (formal) sum of components homogeneous in  bi-weight $\w = (\w_{1},\w_{2})= (1-n, n)$, for $n \geq 0$. The second weight $\w_{2}$ gives the ``tensor order" of the component  and thus describes the \emph{arity} of the bracket associated with that component. The first weight $\w_{1}$ gives the \emph{weight} of the bracket associated with that component. That is the $n$-aray bracket  on $C^{\infty}(\Pi E^{*})$ or $C^{\infty}(E^{*})$ is of weight $(1-n)$. Recall  that the algebras $C^{\infty}(\Pi E^{*})$ and $C^{\infty}(E^{*})$ naturally carry the weight associated with the bundle vector bundle structure $E^{*}\rightarrow M$.  \\

\noindent  The higher Schouten brackets on $C^{\infty}(\Pi E^{*})$, that is ``vector bundle multivector fields" are provided by:

\begin{equation}
(X_{1}, X_{2}, \cdots , X_{r})_{S} = \left.\{ \cdots\{ \{S, X_{1} \}, X_{2}, \cdots \} X_{r}  \}\right|_{\Pi E^{*} \subset T^{*}(\Pi E^{*})},
\end{equation}

\noindent where $X_{I} \in C^{\infty}(\Pi E^{*})$ and the brackets are canonical Poisson brackets on $T^{*}(\Pi E^{*})$.    \\

\noindent Similarly, the higher Poisson brackets on $C^{\infty}(E^{*})$, that is ``vector bundle  symmetric contravariant tensors" are provided by:

\begin{equation}
\{F_{1}, F_{2}, \cdots , F_{r}  \}_{P} = (-1)^{\varepsilon}\left. \SN{ \cdots \SN{ \SN{P, F_{1}} , F_{2}},  \cdots ,  F_{r}  }   \right|_{E^{*} \subset \Pi T^{*}(E^{*})},
\end{equation}

\noindent where $F_{I} \in C^{\infty}(E^{*})$ and the brackets are the canonical Schouten--Nijenhuist brackets on $\Pi T^{*}(E^{*})$. The sign factor is given by $ \varepsilon = (\widetilde{F}_{1} (r-1) + \widetilde{F}_{2}(r-2) + \cdots + \widetilde{F}_{r-1} +r )$\\

\noindent We are now in a position to state a few direct corollaries to Theorem(\ref{theorem 1}).

\begin{corollary}{Corollary}\label{corollary1}
For a  strict $L_{\infty}$-algebroid $(\Pi E, Q)$, the associated higher Schouten and  higher Poisson algebras on $C^{\infty}(\Pi E^{*})$ and $C^{\infty}(E^{*})$ are as  $L_{\infty}$-algebras both strict.
\end{corollary}
\noindent  In terms of the higher Schouten and higher Poisson structures themselves, this translates to the condition $S|_{\Pi E^{*} \subset T^{*}(\Pi E^{*})}=0$ and $P|_{E^{*} \subset \Pi T^{*}(E^{*})}=0$. This is clear from counting the weight(s) or just examining the local expressions. This justifies our nomenclature. Note that Lie algebroids give rise to ``classical-binary"  Schouten and Poisson structures which are clearly strict as higher structures.\\

\noindent By thinking of  $L_{\infty}$-algebras to be $L_{\infty}$-algebroids over a ``point" we  arrive at another corollary.

\begin{corollary}{Corollary}\label{corollary2}
An  $L_{\infty}$-algebra $(U, \{, \cdots ,\})$ is equivalent to:
\begin{enumerate}
\item a homological vector field $Q \in \Vect(\Pi U)$.
\item a  homotopy Schouten algebra on $C^{\infty}(\Pi U^{*})$, with the $n$-th bracket of natural weight $(1-n)$.
\item a  homotopy Poisson algebra on $C^{\infty}(U^{*})$, with the  $n$-th bracket of natural weight $(1-n)$.
\end{enumerate}
If the $L_{\infty}$-algebra is strict,  $Q$ vanishes at the origin, the associated homotopy Schouten and homotopy Poisson algebras are as $L_{\infty}$-algebras both strict.
\end{corollary}

\noindent The above directly generalises what is known about Lie algebras. These  higher Schouten and Poisson brackets are considered to be the homotopy generalisation of the Lie--Schouten and Lie--Poisson bracket.  To the authors knowledge, this association of   homotopy Schouten and homotopy Poisson algebras with  general $L_{\infty}$-algebras  has not appeared  in the literature before.    \\

\noindent Let us be more explicit here. Let the homological vector field describing an arbitrary $L_{\infty}$-algebra be given by
\begin{equation}
Q = \left(  Q^{\delta} +  \xi^{\alpha} Q_{\alpha}^{\delta} + \frac{1}{2!} \xi^{\alpha} \xi^{\beta} Q_{\beta \alpha}^{\delta} + \frac{1}{3!} \xi^{\alpha}\xi^{\beta}\xi^{\gamma}Q_{\gamma \beta \alpha}^{\delta} + \cdots \right)\frac{\partial}{\partial \xi^{\delta}} \in \Vect(\Pi U).
\end{equation}

\noindent Then picking an ``odd" basis $s_{\alpha} \in \Pi U$, $\widetilde{s}_{\alpha} = \widetilde{\alpha} +1$ the symmetric brackets on $\Pi U$ are given by

\begin{equation}
(s_{\alpha_{1}}, s_{\alpha_{2}}, \cdots , s_{\alpha_{r}}) = (-1)^{(\sum_{i=1}^{r}\widetilde{\alpha}_{i}) } Q^{\beta}_{\alpha_{1} \alpha_{2} \cdots \alpha_{r}} s_{\beta}.
\end{equation}

\noindent The associated higher Lie--Schouten brackets on $C^{\infty}(\Pi U^{*})$ are also symmetric and given by

\begin{equation}
(X_{1}, X_{2}, \cdots , X_{r})_{S} = (-1)^{\epsilon} Q^{\beta}_{\alpha_{r} \cdots \alpha_{2}  \alpha_{1}} \eta_{\beta} \frac{\partial X_{1}}{\partial \eta_{\alpha_{1}} } \frac{\partial X_{2}}{\partial \eta_{\alpha_{2}} } \cdots \frac{\partial X_{r}}{\partial \eta_{\alpha_{r}} },
\end{equation}

\noindent the sign factor is given by
\begin{eqnarray}
\epsilon &=& \widetilde{X_{1}}\left( \widetilde{\alpha_{2}} + \widetilde{\alpha_{3}} + \cdots \widetilde{\alpha_{r}}+ r+1 \right) \\
\nonumber &+& \widetilde{X_{2}}\left( \widetilde{\alpha_{3}} + \widetilde{\alpha_{4}} + \cdots \widetilde{\alpha_{r}}+ r+2 \right)\\
\nonumber &+& \widetilde{X_{3}}\left( \widetilde{\alpha_{4}} + \widetilde{\alpha_{5}} + \cdots \widetilde{\alpha_{r}}+ r+3 \right)\\
\nonumber &\vdots& \vdots \\
\nonumber &+&\widetilde{X}_{r-2}\left( \widetilde{\alpha}_{r-1} + \widetilde{\alpha_{r}} -2 \right)\\
\nonumber &+& \widetilde{X}_{r-1}\left( \widetilde{\alpha_{r}} -1 \right)\\
\nonumber &+ & \widetilde{\alpha_{1}} + \widetilde{\alpha_{2}} + \cdots \widetilde{\alpha_{r}}.
\end{eqnarray}
Specifically, the \emph{fundamental Lie--Schouten brackets} are given by

\begin{equation}
(\eta_{\alpha_{1}}, \eta_{\alpha_{2}}, \cdots, \eta_{\alpha_{r}})_{S} = (-1)^{\sum_{i=1}^{r}\widetilde{\alpha_{i}}}Q_{\alpha_{1} \alpha_{2} \cdots \alpha_{r}}^{\beta}\eta_{\beta}.
\end{equation}

\noindent Alternatively one can consider brackets on $U$ that are skew-symmetric. By picking a basis $T_{\alpha} (= \Pi s_{\alpha})$ the are given by

\begin{equation}
\{T_{\alpha_{1}}, T_{\alpha_{2}}, \cdots , T_{\alpha_{r}}  \} = (-1)^{(\sum_{i=1}^{r} \widetilde{\alpha}_{i} (r-i+1) +1) }Q^{\beta}_{\alpha_{1} \alpha_{2} \cdots \alpha_{r}}T_{\beta}.
\end{equation}

\noindent The associated higher Lie--Poisson brackets on $C^{\infty}(U^{*})$ are anti-symmetric and given by

\begin{equation}
\{F_{1}, F_{2}, \cdots, F_{r}  \}_{P} = (-1)^{\varepsilon} Q^{\beta}_{\alpha_{r} \cdots \alpha_{2}  \alpha_{1}} e_{\beta} \frac{\partial F_{1}}{\partial e_{\alpha_{1}}}\frac{\partial F_{2}}{\partial e_{\alpha_{2}}} \cdots \frac{\partial F_{r}}{\partial e_{\alpha_{r}}},
\end{equation}

\noindent the sign factor here is given by

\begin{eqnarray}
 \varepsilon &=& (\widetilde{F}_{1} + 1)(\widetilde{\alpha}_{2} + \widetilde{\alpha}_{3} + \cdots \widetilde{\alpha}_{r}+ r+1)\\
\nonumber &+& (\widetilde{F}_{2} + 1)(\widetilde{\alpha}_{3} + \widetilde{\alpha}_{4} + \cdots\widetilde{\alpha}_{r}+ r+2)\\
\nonumber &\vdots&  \vdots \\
\nonumber &+&(\widetilde{F}_{r-1} + 1)\widetilde{\alpha}_{r}\\
\nonumber &+&\widetilde{F}_{1}(r-1) + \widetilde{F}_{2}(r-1) + \cdots \widetilde{F}_{r-1}\\
\nonumber &+&\widetilde{\alpha}_{1} + \widetilde{\alpha}_{2} + \cdots + \widetilde{\alpha}_{r}+1.
\end{eqnarray}

\noindent Specifically, the \emph{fundamental Lie--Poisson brackets} are given by

\begin{equation}
\{e_{\alpha_{1}}, e_{\alpha_{2}}, \cdots , e_{\alpha_{r}}  \}_{P} = (-1)^{(\sum_{i=1}^{r} \widetilde{\alpha}_{i} (r-i+1) +1) }Q^{\beta}_{\alpha_{1} \alpha_{2} \cdots \alpha_{r}}e_{\beta}.
\end{equation}

 \begin{statement}
 The higher Schouten and Poisson algebras contain the original $L_{\infty}$-algebra. More correctly:\\
 \begin{itemize}
 \item Let us view the vector space $\Pi U \subset C^{\infty}(\Pi U^{*})$ as the weight one functions. Then the $L_{\infty}$-algebra  brackets on $\Pi U$ are exactly given by the restriction of the Lie--Schouten brackets to weight one functions.
 \item Let us view the vector space $U \subset C^{\infty}(U^{*})$ as weight one functions. Then the $L_{\infty}$-algebra brackets on $U$ are exactly given by the restriction of the Lie--Poisson brackets to weight one functions
 \end{itemize}
 \end{statement}

\noindent It is largely a  matter of taste if one wishes to work with ``odd" or ``even" structures when delating with $L_{\infty}$-algebras in the $\mathds{Z}_{2}$-graded setting. However, there is in general less sign factors to handle when working with odd structures. Furthermore, due to the relation with Voronov's higher derived bracket formalism  it seems very natural to consider odd symmetric brackets and higher Schouten structures   as being in some sense primitive or fundamental.

\begin{remark}
The association of an homotopy Schouten algebra with an  $L_{\infty}$-algebra opens up the possibility of describing $L_{\infty}$-bialgebras following the recipe of Roytenberg \cite{roytenberg-1999} and Voronov \cite{Voronov:2001qf} who develop the theory of Lie bialgebroids. A little more specifically, one can define the notion of an  $L_{\infty}$-bialgebra as an $L_{\infty}$-algebra $(\Pi U, Q_{U})$ together with an $L_{\infty}$-algebra on the dual space $(\Pi U^{*}, Q_{U^{*}})$ such that the homological vector field $Q_{U}$ satisfies a Leibnitz rule over the higher Schouten brackets on $C^{\infty}(\Pi U)$.  The notion of a homotopy version of a Lie bialgebra can be traced back to the work of Kravchenko \cite{kravchenko-2006}, in a $\mathds{Z}$-graded setting. It would be very desirable to understand the details of how the original constructions of Kravchenko relate to that suggested here.   We hope to present details elsewhere.
\end{remark}

\section{Simple Examples}\label{simpleexamples}
\noindent Let us present a few simple examples to help clarify Theorem(\ref{theorem 1}). In order to keep this section relatively simple and self-contained we will concentrate on low order structures.  Where appropriate we direct the reader to the original literature for further details.

\begin{example}{Example} \label{ex1}\textbf{The de Rham differential and canonical structures}\\
\noindent Consider the tangent bundle of a manifold $TM$.  The relevant homological vector field is the de Rham differential. The associated brackets are the canonical Schouten--Nijenhuist bracket on $\Pi T^{*}M$ and the canonical Poisson bracket on $T^{*}M$. The de Rham differential is of weight one and the Schouten/Poisson structures are of bi-weight $(-1, 2)$.
\begin{subequations}
\begin{eqnarray}
Q &=& d = dx^{A}\frac{\partial}{\partial x^{A}} \in \Vect(\Pi TM),\\
S &=&  (-1)^{\widetilde{A}} \pi^{A}p_{A} \in C^{\infty}(T^{*}(\Pi T^{*}M)),\\
P &=&  e_{*}^{A} x^{*}_{A}\in C^{\infty}(\Pi T^{*}(T^{*}M)).
\end{eqnarray}
\end{subequations}
 \end{example}

\begin{example}{Example}\label{ex2} \textbf{Lie algebroids}\\
\noindent By  concentrating on  $n=2$ one naturally recovers Lie algebroids.

\begin{subequations}
\begin{eqnarray}
Q &=& \xi^{\alpha}Q_{\alpha}^{A} \frac{\partial}{\partial x^{A}} + \frac{1}{2!} \xi^{\alpha}\xi^{\beta}Q_{\beta \alpha}^{\gamma}\frac{\partial}{\partial \xi^{\gamma}} \in \Vect(\Pi E),\\
S &=& (-1)^{\widetilde{\alpha}}\pi^{\alpha}Q_{\alpha}^{A} p_{A} + (-1)^{\widetilde{\alpha} + \widetilde{\beta}}\frac{1}{2!} \pi^{\alpha}\pi^{\beta}Q_{\beta \alpha}^{\gamma}\eta_{\gamma} \in C^{\infty}(T^{*}(\Pi E^{*})),\\
P &=& e_{*}^{\alpha}Q_{\alpha}^{A} x^{*}_{A} - \frac{1}{2!} e_{*}^{\alpha}e_{*}^{\beta}Q_{\beta \alpha}^{\gamma}e_{\gamma}\in C^{\infty}(\Pi T^{*}(E^{*})).
\end{eqnarray}
\end{subequations}

\noindent Note that for a Lie algebroid the homological vector field $Q \in \Vect(\Pi E)$ is of weight one and that the Schouten and Poisson structures are of bi-weight $(-1,2)$. Naturally taking the base manifold to be a ``point" one recovers Lie algebras. \\
\end{example}

\begin{example}{Example}\label{ex3}\textbf{Lie 3-algebroids}\\
\noindent A  Lie 3-algebroid is an $L_{\infty}$-algebroid $(\Pi E,Q)$ such that the homological vector field is concentrated  in weight from minus one up to  and including weight two.
\begin{eqnarray}
\nonumber Q &=& \left( Q^{A} + \xi^{\alpha}Q_{\alpha}^{A} + \frac{1}{2!} \xi^{\alpha} \xi^{\beta} Q_{\beta \alpha}^{A} \right) \frac{\partial}{\partial x^{A}}\\
&+& \left(  Q^{\delta} +  \xi^{\alpha} Q_{\alpha}^{\delta} + \frac{1}{2!} \xi^{\alpha} \xi^{\beta} Q_{\beta \alpha}^{\delta} + \frac{1}{3!} \xi^{\alpha}\xi^{\beta}\xi^{\gamma}Q_{\gamma \beta \alpha}^{\delta}\right)\frac{\partial}{\partial \xi^{\delta}} \in \Vect(\Pi E).
\end{eqnarray}

\noindent   The associated higher Schouten and higher Poisson structures are given by
\begin{subequations}
\begin{eqnarray}
S &=& \left(Q^{A} + (-1)^{\widetilde{\alpha}} \pi^{\alpha} Q_{\alpha}^{A} + (-1)^{\widetilde{\alpha} + \widetilde{\beta}} \frac{1}{2!} \pi^{\alpha} \pi^{\beta}Q_{\beta \alpha}^{A}  \right)p_{A}\\
\nonumber &+&\left(  Q^{\delta} +(-1)^{\widetilde{\alpha}}\pi^{\alpha} Q_{\alpha}^{\delta} + (-1)^{\widetilde{\alpha} + \widetilde{\beta}}\frac{1}{2!} \pi^{\alpha} \pi^{\beta} Q_{\beta \alpha}^{\delta} \right.\\
\nonumber &+& \left. (-1)^{\widetilde{\alpha}+ \widetilde{\beta} +\widetilde{\gamma}}\frac{1}{3!} \pi^{\alpha}\pi^{\beta}\pi^{\gamma}Q_{\gamma \beta \alpha}^{\delta}\right)\eta_{\delta} \in C^{\infty}(T^{*}(\Pi E^{*})).\\
P &=& \left( Q^{A} + e_{*}^{\alpha} Q_{\alpha}^{A} + \frac{1}{2!} e_{*}^{\alpha} e_{*}^{\beta} Q_{\beta \alpha}^{A} \right) x^{*}_{A}\\
\nonumber &-&\left( Q^{\delta}+ e_{*}^{\alpha} Q_{\alpha}^{\delta} + \frac{1}{2!} e_{*}^{\alpha} e_{*}^{\beta} Q_{\beta \alpha}^{\delta} + \frac{1}{3!} e_{*}^{\alpha}e_{*}^{\beta}e_{*}^{\gamma}Q_{\gamma \beta \alpha}^{\delta}\right)e_{\delta} \in C^{\infty}(\Pi T^{*}(E^{*})).
\end{eqnarray}
\end{subequations}
\noindent Note the higher structures consist of the sum of bi-weight $(1,0)$, $(0,1)$, $(-1,2)$ and $(-2,3)$ terms. Thus, the homotopy Schouten/Poisson algebras consist of 0-aray, 1-aray, 2-aray and 3-aray brackets. Taking the base manifold to be a ``point" one is lead to what is known as a Lie 3-algebra.
\end{example}

\newpage

\begin{example}{Example}\label{ex4}\textbf{Graded 3-Lie algebras}\\
\noindent Let $U$ be a (super)vector space. A graded (or super) 3-Lie algebra is taken to be an $L_{\infty}$-algebra such that the homological vector field $Q \in \Vect(\Pi U)$ is concentrated in weight two.
\begin{equation}
Q = \frac{1}{3!}\xi^{\alpha}\xi^{\beta}\xi^{\gamma}Q_{\gamma \beta \alpha}^{\delta}\frac{\partial}{\partial \xi^{\delta}} \in  \Vect(\Pi U).
\end{equation}

\noindent The associated higher Schouten and higher Poisson structures of bi-weight $(-2,3)$
\begin{subequations}
\begin{eqnarray}
S &=& (-1)^{\widetilde{\alpha}+ \widetilde{\beta} + \widetilde{\gamma}}\frac{1}{3!} \pi^{\alpha}\pi^{\beta}\pi^{\gamma}Q_{\gamma \beta \alpha}^{\delta}\eta_{\delta} \in C^{\infty}(T^{*}(\Pi U^{*})).\\
P &= & - \frac{1}{3!} e^{\alpha}_{*}e^{\beta}_{*}e^{\gamma}_{*}Q_{\gamma \beta \alpha}^{\delta } e_{\delta} \in C^{\infty}(\Pi T^{*}(U^{*})).
\end{eqnarray}
\end{subequations}
\noindent The homotopy Schouten and Poisson algebras consist of a single ternary bracket. Note that this is not the same as a Lie $3$-algebroid over a ``point", i.e. a Lie 3-algebra which  consists of a series of brackets for $n=0,1,2,3$.
\end{example}

\begin{small}
\begin{aside}
 If the $L_{\infty}$-algebra has only a non-vanishing $n$-aray bracket then one  has a \emph{graded $n$-Lie algebra}.  These are not quite the same as Filippov's  $n$-Lie algebras \cite{Filippov:1985} due to the underlying gradings of weight and Grassmann parity.  Such algebras are described by a weight $(n-1)$ homological vector field. The associated higher Schouten/Poisson structures are of bi-weight $(1-n,n)$.  The Bagger--Lambert--Gustavsson model \cite{Bagger:2006sk,Bagger:2007jr,Gustavsson:2007vu} (plus many other references) of multiple coincident M2-branes   is constructed using  (metric) $3$-Lie algebras.   Reformulating the BLG-model and the generalised Nham equation of Basu \& Harvey \cite{Basu:2004ed} in the language of $L_{\infty}$-algebras was undertaken by Iuliu-Lazaroiu et al \cite{IuliuLazaroiu:2009wz}. Also Lambert \& Papageorgakis \cite{Lambert:2010wm} very recently provided evidence that $3$-aray algebras are also fundamental in the effective description of M5 branes. There seems to be some deep link between M-theory and $n$-aray algebras.   Thus, it is natural to wonder if any of  the work presented in this paper is of any relevance here. A little care over the sign factors appearing in the constructions would be required. However, this should be very tractable.
\end{aside}
\end{small}

\begin{example}{Example} \label{ex5}\textbf{Higher Poisson structures on Lie algebroids}\\
\noindent This example is taken  directly from \cite{Bruce-2010},  also see \cite{khudaverdian-2008} for the specific case where $E = TM$. Recall that a Lie algebroid $E \rightarrow M$ is completely encoded in a weight minus one Schouten bracket on $C^{\infty}(\Pi E^{*})$, see Example(\ref{ex2}). Let us denote this Schouten bracket as $\SN{\bullet, \bullet}_{E}$. A higher Poisson structure on the Lie algebroid $E$ is defined as an even parity, but generally inhomogeneous in weight  ``multivector" $\mathcal{P} \in C^{\infty}(\Pi E^{*})$ such that $\SN{\mathcal{P},\mathcal{P}}_{E}=0$. Associated with such a structure is an homotopy Poisson algebra over the base manifold. The vector bundle $E^{*} \rightarrow M$ comes equipped with the structure of an $L_{\infty}$-algebroid viz $Q_{\mathcal{P}} = - \SN{\mathcal{P}, \bullet}_{E} \in \Vect(\Pi E^{*})$.\\

\noindent Then via Theorem(\ref{theorem 1}) we see that
\begin{enumerate}
\item The algebra of ``differential forms" $C^{\infty}(\Pi E)$ comes equipped with an homotopy Schouten algebra.
\item The algebra of ``symmetric tensors" $C^{\infty}(E)$ comes equipped with an homotopy Poisson algebra.
\end{enumerate}
\end{example}

\section{Concluding remarks}\label{concluding remarks}

\noindent We must remark that we have worked in the ``super-setting" and that the (bi-)weight attached to the coordinates and the brackets keep track of the ``algebra" in a geometric way. Although there is no canonical choice of weights, the ones used here seem quite natural as far as the geometry is concerned. For the specific constructions relating to $L_{\infty}$-algebras, it is possible to amend the constructions presented in this work to be inline with the original gradings of Lada \& Stasheff  \cite{Lada:1994mn,Lada:1992wc}.\\

\noindent Voronov \cite{Voronov2010B} recently defined  \emph{non-linear Lie algebroids}\footnote{In fact in Voronov in \cite{Voronov:2001qf} defines \emph{non-linear Lie bialgebroids}. } as homological vector fields of weight one over non-negatively graded manifolds. It is natural to consider inhomogeneous homological vector fields over a non-negatively graded manifold as a further generalisation of $L_{\infty}$-algebroids.      \\

\noindent However, it is not  clear how the constructions presented here would carry over to homological vector fields over  non-negatively  graded manifolds, more general graded manifolds or even derived manifolds. The graded structures associated with  vector and double vector bundle structures feature prominently. The presence of ``linear structures" is essential in dualising, applying the parity reversion functor and quite critically in employing the double vector bundle morphisms.  One other direction is to consider higher vector bundles and ``higher Legendre transformations".  Multigraded manifolds can be employed to set-up the theory of higher vector bundles,  see for example \cite{Grabowski2009,Voronov:2001qf,voronov-2006}. This sits comfortably with the constructions presented here.  The notion of multiple $L_{\infty}$-algebroids would require a good understanding of such structures over multigraded manifolds.  \\

\noindent The relation between $L_{\infty}$-algebroids (as defined here) and the BV and BFV formalisms is not completely transparent. For example, very similar structures can be found when generalising the BV formalism to higher or nonabelian antibrackets \cite{Alfaro:1995vw,Batalin:1998cz,Batalin:1999gf,Bering:1996kw}. A clear and precise understanding of the role of $L_{\infty}$-algebroids  in quantum field theory  would be very desirable. However, generally further weights including non-positive weights would be required for a complete understanding. Thus, geometric structures over derived manifolds are likely to be key in uncovering the relevance of Theorem(\ref{theorem 1})  in quantum field theory \cite{Schreiber2010}.  In any case, it is tempting to think of $L_{\infty}$-algebroid structures (Eqn.(\ref{homologicalfieldLieinfity})) as some kind of  inhomogeneous  higher ghost number BRST operator akin to that found in \cite{Chryssomalakos:1998gh,DeAzcarraga:1996ts}. The physical or mathematical  relevance of such  operators is not immediately obvious,  further study is required.\\

\noindent  At present  higher Poisson and Schouten structures over supermanifolds or on Lie algebroids  provide the most natural non-trivial examples of $L_{\infty}$-algebroids. With hindsight this was to be expected given the intimate relation between Lie algebroids and Poisson geometry. Theorem(\ref{theorem 1}) together with results found in \cite{Bruce-2010,khudaverdian-2008} demonstrate that this relation passes over to their homotopy relatives without substantial effort.  The key to this relative ease is to employ  graded manifolds. \\

\section*{Acknowledgments}

\noindent The author would like to thank   D. Roytenberg, U. Schreiber and J. Stasheff for their invaluable insight. The author must also  thank R. Szabo for pointing out references  to M-theory.  A special thank you must go to the anonymous referee whose comments and suggestions served to vastly  improve the presentation of  this work. \\

\appendix
\section*{Appendix}\label{appendix}
\section{Canonical double vector bundle morphisms}
\noindent For completeness we present the canonical double vector bundle morphisms used in the proof of the main theorem. In particular we prove that the morphisms are symplectomorphisms.  Further details and application of these morphisms can be found in \cite{Bruce-2010,roytenberg-1999,Voronov:2001qf}.

\subsection{$T^{*}(\Pi E^{*})$ and $T^{*}(\Pi E)$}\label{A1}
\noindent Let us employ natural local coordinates:

\vspace{15pt}
\begin{tabular}{|l ||l| }
\hline
$ T^{*}(\Pi E^{*})$ & $\{x^{A},\eta_{\alpha}, p_{A}, \pi^{\alpha} \}$ \\
$  T^{*}(\Pi E)$  & $\{x^{A},\xi^{\alpha}, p_{A}, \pi_{\alpha} \}$\\
\hline
\end{tabular}\\

\noindent The parities are given by $ \widetilde{x}^{A}= \widetilde{p}_{A}= \widetilde{A}$, $\widetilde{\eta}_{\alpha}= \widetilde{\pi}^{\alpha}= \widetilde{\pi}_{\alpha}= \widetilde{\xi}^{\alpha} =  \widetilde{\alpha}+1$. The bi-weights are assigned as $\w(x^{A}) =(0,0)$, $\w(\eta_{\alpha}) =(1,0)$, $\w(p_{A}) =(0,1)$, $\w(\pi^{\alpha}) =(-1,1)$ , $\w(\xi^{\alpha}) =(-1,1)$, $\w(\pi_{\alpha}) =(1,0)$.\\

\noindent The admissible changes of coordinates are:

\vspace{15pt}
\begin{tabular}{|l||l|}
\hline
$T^{*}(\Pi E^{*})$ &   $\overline{x}^{A}  =  \overline{x}^{A}(x)$, \hspace{5pt} $\overline{\eta}_{\alpha}  =   (T^{-1})_{\alpha}^{\:\: \beta}\eta_{\beta}$,\\
 & $\overline{p}_{A} = \left( \frac{\partial x^{B}}{\partial \overline{x}^{A}} \right)p_{B} + (-1)^{\widetilde{A}(\widetilde{\gamma}+ 1) + \widetilde{\delta}} \pi^{\delta}T_{\delta}^{\:\: \gamma} \left( \frac{\partial (T^{-1})_{\gamma}^{\:\: \alpha}}{\partial \overline{x}^{A}} \right)\eta_{\alpha}$,\\
 & $ \overline{\pi}^{\alpha} = (-1)^{\widetilde{\alpha} + \widetilde{\beta}}\pi^{\beta}T_{\beta}^{\:\: \alpha}$.\\
\hline
$T^{*}(\Pi E)$ & $ \overline{x}^{A}  =  \overline{x}^{A}(x)$, \hspace{5pt}$\overline{\xi}^{\alpha}  =   \xi^{\beta} T_{\beta}^{\:\: \alpha}$,\\
& $\overline{p}_{A} = \left( \frac{\partial x^{B}}{\partial \overline{x}^{A}} \right)p_{B} + (-1)^{\widetilde{A}(\widetilde{\gamma}+1)} \xi^{\delta}T_{\delta}^{\:\: \gamma} \left(\frac{\partial (T^{-1})_{\gamma}^{\:\: \alpha}}{\partial \overline{x}^{A}}  \right)\pi_{\alpha}$,\\
&  $\overline{\pi}_{\alpha} = (T^{-1})_{\alpha}^{\:\: \beta} \pi_{\beta}$.\\
\hline
\end{tabular}\\

\vspace{15pt}

\noindent There is canonical double vector bundle morphism $R: T^{*}(\Pi E^{*}) \rightarrow T^{*}(\Pi E )$ given in local coordinates as

\begin{equation}
R^{*}(\pi_{\alpha}) = \eta_{\alpha},  \hspace{35pt} R^{*}(\xi^{\alpha}) = (-1)^{\widetilde{\alpha}}\pi^{\alpha}.
\end{equation}

\begin{lemma}{Lemma}\label{lemma1}
The canonical double vector bundle morphism $R: T^{*}(\Pi E^{*}) \rightarrow T^{*}(\Pi E )$ is a symplectomorphism.
\end{lemma}

\begin{proof}
The canonical even  symplectic structure on $T^{*}(\Pi E^{*})$ is given by $\omega_{T^{*}(\Pi E^{*})} =  dp_{A}dx^{A} + d\pi^{\alpha}d \eta_{\alpha}$ and on $T^{*}(\Pi E)$ is given by $\omega_{T^{*}(\Pi E)} = dp_{A}dx^{A} + d\pi_{\alpha}d\xi^{\alpha}$. Thus, $R^{*}\omega_{T^{*}(\Pi E )} = \omega_{T^{*}(\Pi E^{*})}$ and we see that $R$ is indeed a symplectomorphism.\\
\end{proof}

\subsection{$\Pi T^{*}(E^{*})$ and $\Pi T^{*}(\Pi E)$}\label{A2}
\noindent Let us employ natural local coordinates:

\vspace{15pt}
\begin{tabular}{|l ||l| }
\hline
$\Pi T^{*}(E^{*})$ & $\{x^{A},e_{\alpha}, x^{*}_{A}, e_{*}^{\alpha} \}$\\
$ \Pi T^{*}(\Pi E)$  & $\{x^{A},\xi^{\alpha}, x^{*}_{A}, \xi_{\alpha}^{*} \}$\\
\hline
\end{tabular}\\

\noindent The parities are given by  $\widetilde{x}^{A} = \widetilde{A}$, $\widetilde{e}_{\alpha}= \xi^{*}_{\alpha}= \widetilde{\alpha}$, $\widetilde{x}^{*}_{A} = \widetilde{A}+1$, $\widetilde{\xi}^{\alpha} = \widetilde{e}_{*}^{\alpha}= \widetilde{\alpha}+1$.  The bi-weights are assigned as $\w(x^{A}) =(0,0)$, $\w(e_{\alpha}) =(1,0)$, $\w(x^{*}_{A}) =(0,1)$, $\w(e_{*}^{\alpha}) =(-1,1)$, $\w(\xi^{\alpha}) =(-1,1)$, $\w(\xi_{\alpha}^{*}) =(1,0)$.\\

\noindent The admissible changes of coordinates are:

\vspace{15pt}
\begin{tabular}{|l||l|}
\hline
$\Pi T^{*}( E^{*})$ &   $\overline{x}^{A}  =  \overline{x}^{A}(x)$, \hspace{5pt} $\overline{e}_{\alpha}  =   (T^{-1})_{\alpha}^{\:\: \beta}e_{\beta}$,\\
 & $\overline{x}^{*}_{A} = \left( \frac{\partial x^{B}}{\partial \overline{x}^{A}} \right)x^{*}_{B} - (-1)^{\widetilde{A}(\widetilde{\gamma}+ 1) + \widetilde{\delta}} e_{*}^{\delta}T_{\delta}^{\:\: \gamma} \left( \frac{\partial (T^{-1})_{\gamma}^{\:\: \alpha}}{\partial \overline{x}^{A}} \right)e_{\alpha}$,\\
 & $ \overline{e}_{*}^{\alpha} = e_{*}^{\beta}T_{\beta}^{\:\: \alpha}$.\\
\hline
$\Pi T^{*}(\Pi E)$ & $ \overline{x}^{A}  =  \overline{x}^{A}(x)$, \hspace{5pt}$\overline{\xi}^{\alpha}  =   \xi^{\beta} T_{\beta}^{\:\: \alpha}$,\\
& $\overline{x}^{*}_{A} = \left( \frac{\partial x^{B}}{\partial \overline{x}^{A}} \right)x^{*}_{B} + (-1)^{\widetilde{A}(\widetilde{\gamma}+1)} \xi^{\delta}T_{\delta}^{\:\: \gamma} \left(\frac{\partial (T^{-1})_{\gamma}^{\:\: \alpha}}{\partial \overline{x}^{A}}  \right)\xi^{*}_{\alpha}$,\\
&  $\overline{\xi}^{*}_{\alpha} = (T^{-1})_{\alpha}^{\:\: \beta} \xi^{*}_{\beta}$.\\
\hline
\end{tabular}\\
\vspace{15pt}

\noindent There is a  canonical double vector bundle morphism $R:\Pi T^{*}(E^{*}) \rightarrow \Pi T^{*}(\Pi E)$  given in local coordinates as

\begin{equation}
R^{*}(\xi^{\alpha}) = e_{*}^{\alpha}, \hspace{30pt} R^{*}(\xi_{\alpha}^{*}) = -e_{\alpha}.
\end{equation}

\begin{lemma}{Lemma}\label{lemma2}
The canonical double vector bundle morphism $R:\Pi T^{*}(E^{*}) \rightarrow \Pi T^{*}(\Pi E)$ is an odd symplectomorphism.
\end{lemma}
\begin{proof}
The canonical odd symplectic structures are given by $\omega_{\Pi T^{*}(E^{*})} = (-1)^{\widetilde{A}+1} dx^{*}_{A}dx^{A} + (-1)^{\widetilde{\alpha}+1} de_{*}^{\alpha}de_{\alpha}$ and $\omega_{\Pi T^{*}(\Pi E)} = (-1)^{\widetilde{A}+1} dx^{*}_{A}dx^{A} + (-1)^{\alpha} d \xi_{\alpha}^{*}d \xi^{\alpha}$. Thus $R^{*}\omega_{\Pi T^{*}(\Pi E )} = \omega_{T^{*}(E^{*})}$ and we see that $R$ is indeed an odd symplectomorphism.\\
\end{proof}


\end{document}